\documentclass[11pt,letterpaper]{article}

\usepackage[margin=1in]{geometry}
\usepackage{microtype}
\usepackage{graphicx}
\usepackage[font=footnotesize]{caption}
\usepackage{booktabs} % for professional tables
\usepackage[normalem]{ulem}
\usepackage[symbol]{footmisc}
\usepackage[hidelinks]{hyperref}

\usepackage{natbib}

\usepackage{amsmath}
\usepackage{amssymb}
\usepackage{mathtools}
\usepackage{amsthm}
\usepackage{subcaption,nicefrac}

\usepackage{xspace}

\usepackage{thmtools} 
\usepackage{thm-restate}

\usepackage[capitalize,noabbrev]{cleveref}

\usepackage{algorithm}
\usepackage[noend]{algorithmic}

\theoremstyle{plain}
\newtheorem{theorem}{Theorem}[section]

\newtheorem{claim}[theorem]{Claim}
\newtheorem{lemma}[theorem]{Lemma}

\theoremstyle{definition}
\newtheorem{definition}[theorem]{Definition}

\theoremstyle{remark}

\newtheorem{example}[theorem]{Example}

\newcommand{\bbRp}{{\mathbb{R}_{\ge 0}}}
\newcommand{\rb}[1]{\left( #1 \right)} % round Brackets
\newcommand{\scalar}[2]{\langle {#1} ,{#2}\rangle}
\DeclareMathOperator*{\argmin}{argmin}

\newcommand{\e}{\varepsilon}
\newcommand{\cA}{\mathcal A}

\newcommand{\OPT}{\operatorname{OPT}}
\newcommand{\swapping}{\textsc{Swapping}\xspace}
\newcommand{\sieve}{\textsc{Sieve-Streaming}\xspace}
\newcommand{\sieveplus}{\textsc{Sieve-Streaming++}\xspace}
\newcommand{\window}{\textsc{Encompassing-Set}\xspace}
\newcommand{\eswapping}{\textsc{Chasing-Local-Opt}\xspace}
\newcommand{\swap}{\textsc{Min-Swap}\xspace}

\newcommand{\marginal}[2]{f\rb{#1 \mid #2}}

\title{Consistent Submodular Maximization}

\author{
	Paul D{\"u}tting\thanks{Google Research.}
	\and
	Federico Fusco\thanks{Department of Computer, Control, and Management 
Engineering, Sapienza University of Rome, 
Italy.}
\and
Silvio Lattanzi{$^*$}
\and
Ashkan Norouzi-Fard{$^*$}
\and
Morteza Zadimoghaddam{$^*$}}

\date{}

\begin{document}

\maketitle

\begin{abstract}
Maximizing monotone submodular functions under cardinality constraints is a classic optimization task with several applications in data mining and machine learning. In this paper we study this problem in a dynamic environment with consistency constraints: elements arrive in a streaming fashion and the goal is maintaining a constant approximation to the optimal solution while having a stable solution (i.e., the number of changes between two consecutive solutions is bounded).
We provide algorithms in this setting with different trade-offs between consistency and approximation quality. We also complement our theoretical results with an experimental analysis showing the effectiveness of our algorithms in real-world instances.
\end{abstract}
%%% Main Body

% Introduction
\section{Introduction}
\label{sec:intro}

    Submodular optimization is a powerful framework for modeling and solving problems that exhibit the widespread diminishing returns property. Thanks to its effectiveness, it has been applied across diverse domains, including video analysis~\citep{ZhengJCP14}, data summarization~\citep{LinB11,BairiIRB15}, sparse reconstruction~\citep{Bach10,DasK11}, and active learning~\citep{GolovinK11,AmanatidisFLLR22}.
    
    In this paper, we focus on submodular maximization under cardinality constraints: given a submodular function $f$, a universe of elements $V$, and a cardinality constraint $k$, the goal is to find a set $S$ of at most $k$  elements that maximizes $f(S)$. 
    Submodular maximization under cardinality constraints is NP-hard, nevertheless efficient %\pdedit{(constant-factor)}
    approximation algorithms exist for this task in both the centralized and the streaming setting \citep{nemhauser1978analysis,BadanidiyuruMKK14,kazemi2019submodular}.
    
    {One aspect of efficient approximation algorithms for submodular maximization that has received little attention so far, is the stability of the solution.}
    In fact, for some of the known algorithms, even adding a single element to the universe of elements $V$ may completely change the final output (see \Cref{app:instability} for some examples). 
    Unfortunately, this is problematic in many real-world applications where consistency is a fundamental system requirement.
    Indeed, a flurry of recent work has started to explore various optimization problems under stability and consistency constraints such as clustering \citep{lattanzi2017consistent,cohen2022online,fichtenberger2021consistent,guo2021consistent,lkacki2024fully}, facility location \citep{cohen2019fully,bhattacharya2022efficient}, and online learning \citep{jaghargh2019consistent}. 
    
    Having solutions that evolve smoothly is central in many practical application of submodular optimization. 
    Consider, for example, the data summarization task in an evolving setting where elements are added to the universe $V$. In this setting, having a stable summary that changes as little as possible from step to step is very important both for serving the summary to a user or for using it in a machine learning model. In fact, in both settings a drastic change of the solution may have negative impact on system usability, it could harm user attention, and adversely effect the performance of the machine learning model.
    
    For these reasons, in this paper we initiate the study of submodular maximization under consistency constraints, where we allow the solutions to change only slightly after each element insertion. More formally, consider a stream $V$ of exactly $n$ elements, chosen by an 
    % oblivious 
    adversary.
    Denote by $V_t = \{e_1, \ldots, e_t\} \subseteq V$ the set of all elements inserted up to the $t$-th stream operation, and let $\OPT_t$ be an optimum feasible solution for $V_t$.
    Our goal is to design an algorithm with two key properties. 
    On the one hand, we want the algorithm to maintain, at the end of each operation $t$, %$i$,
    a %good 
    solution $S_t \subseteq V_t$, with $|S_t| \le k$, of high value $f(S_t)$. In particular, we say that an algorithm is an $\alpha$-approximation of the best solution if $ \alpha f(S_t) \ge f(\OPT_t)$, for all $t = 1,\dots, n$.
    On the other hand, we want the dynamic solution to not change much after consecutive insertions: we say that an algorithm is $C$-consistent if $|S_t \setminus S_{t-1}| \le C$ for all $t = 2, \dots, n$. In general, we say that an algorithm is consistent, without specifying $C$, when $C$ is constant.
    
    It is interesting to note that the \swapping algorithm by \citet{ChakrabartiK15} already conjugates constant approximation with constant consistency\footnote{Following e.g., \citet{DuettingFLNZ22,DuettingFLNZ23}, we call \swapping the instantiation of the general framework by \citet{ChakrabartiK15} for the special case of matroid constraints. We refer to \Cref{swap:tight} for the pseudocode.}. \swapping maintains a dynamic feasible solution and each new arriving element is added to the solution if either it fits into the cardinality constraint or it is possible to swap it with some low-value element. 
    It is well known that \swapping achieves a $4$-approximation, and from the previous description it is also clear that it is $1$-consistent.%
    %By construction, \swapping is $1$-consistent, and it maintains a $4$-approximation
    \footnote{It is possible to show that the $4$ is tight for the approximation factor. For an example please refer to Appendix~\ref{swap:tight}.}
    
    {Putting consistency aside, it is NP-hard to get an approximation guarantee better than $\nicefrac{e}{(e-1)}$ \citep{Feige98}, which can be achieved by recomputing a greedy solution \citep{nemhauser1978analysis} from scratch after every insertion. However, such approach is not consistent (see \Cref{app:instability})}.

    A line of work that is related to our model is that of fully-dynamic submodular maximization \citep[e.g.,][]{LattanziMNTZ20,Monemizadeh20,DuettingFLNZ23,BanihashemBGHJM24}. There, the algorithm is given an arbitrary stream of insertions and deletions, and the goal is to maintain a good dynamic solution with low amortized running time. While the constraint on running time naturally induces algorithms characterized by solutions that do not change often, known algorithms for fully dynamic submodular maximization are not consistent, as they all contemplate the possibility of recomputing the solution from scratch from time to time.

    \paragraph{Our Contribution.} Given these considerations, it is natural to ask if it is possible to obtain a better trade-off between quality and consistency. We answer this question positively:
    \begin{itemize}
        \item We first provide 
        a $(3.147 +O(\nicefrac 1k))$-approximation
        % $3.147 +O(1/k)$ 
        algorithm that is $1$-consistent,
        improving on the guarantees of the \swapping algorithm.  
        %swapping guarantees.
        \item We then provide a $(2.619 + \e)$-approximation\footnote{As is common in the submodular maximization literature, the parameter $\e$ is intended to be a small constant that the algorithm designer can tune according to the application at hand: it is possible to attain an approximation arbitrarily close to $2.619$, at the cost of a worse consistency.}
        % $2.619 + \epsilon$ 
        algorithm that is $\tilde{O}(\nicefrac 1\e)$-consistent, where the $\tilde{O}$ notation hides poly-logarithmic factors in $\nicefrac 1\e$.
    \end{itemize}
    We complement our positive
    %new 
    results with a lower bound showing that  
    %no deterministic algorithm can be consistent and return a better than $2$ approximation.
    for any constant $C$, no deterministic algorithm can be $C$-consistent and return a better than $2$ approximation. 
    Since both our algorithms are deterministic, the lower bound shows that our algorithms obtain a near-optimal quality-consistency tradeoff. We leave the resolution of the remaining gaps, and the study of randomized algorithms as exciting directions for future work.
    
    We also present extensive experiments with real-world data sets and a synthetic data set 
    (\Cref{sec:experiments,swap:tight,app:other_experiments}). 
    %and a synthetic data set. %(\Cref{swap:tight}) 
    The experiments show that our algorithms achieve comparable value as \swapping and the non-consistent \sieve \citep{BadanidiyuruMKK14} on real-world data sets; while achieving significant savings in the total number of changes. Furthermore, the synthetic data set -- constructed using a hard instance for \swapping presented in \Cref{swap:tight} -- confirms the improvements in the worst-case approximation guarantees relative to \swapping from our theoretical analysis, showing that there too the gains can be significant (in the order of the $21.325\%$ and $34.525\%$ improvements that we show in our analysis).

    \paragraph{Our Techniques.}
    
    Our first algorithm, \window, maintains a benchmark set $B_t$ that is used to decide whether to add or discard any new element. More precisely, any  arriving element $e_t$ is added to $B_{t-1}$ if, upon arrival,the marginal contribution of $e_t$ to $B_{t-1}$, that is $f(e_t \mid B_{t-1})$, is at least $\nicefrac{\beta}{k} \cdot f(B_{t-1})$. Here $\beta$ is a judiciously chosen constant that is larger than $1$, namely $\beta = 1.14$. 
    At any given time $t$, the solution $S_t$ maintained by the algorithm consists of the last $k$ elements added to $B_t$. 
    
    This algorithm is $1$-consistent by construction, while the approximation guarantee 
    %(relative to \swapping) 
    descends from the following two properties of this algorithm. First, the (potentially infeasible) benchmark set $B_t$ achieves a $(1+\beta)$-approximation to $f(\OPT_t)$ (where $1+\beta = 2.24$ by the choice of $\beta$).
    Second, due to the exponential nature of the condition by which elements are added to the benchmark set, the elements in $B_t$ that are not part of $S_t$ only account for a small fraction of the value of $B_t$; namely, $f(B_t) \geq (1+\nicefrac{\beta}{k})^k f(B_t\setminus S_t)$.
    Intuitively, the second property shows that $S_t$ captures a significant fraction of $f(B_t)$, while the first property shows that $f(B_t)$ is a good approximation to $f(\OPT_t)$. A careful analysis shows that the two properties lead to the claimed factor of $3.147 + O(\nicefrac 1k)$.
    
    Our second algorithm, \eswapping, provides a better approximation guarantee at the cost of possibly performing more 
    % and performs possibly more 
    than one swap per step (but still at most constantly many). Rather than maintaining a benchmark set, this algorithm only maintains a solution $S_t$, and updates it via local improvements. 
    It applies a similar swapping condition as \window, by requiring that the marginal value of an arriving element $e_t$ to $S_{t-1}$ should be at least $\nicefrac{\phi}{k} \cdot f(S_{t-1})$, where $\phi \approx 1.61$ is the golden ratio. Other than \window, however, rather than swapping out the oldest element that was added, it swaps out an element $r$ whose marginal contribution ${f(r \mid S - r)}$\footnote{We use $S - r$ instead of $S\setminus \{r\}$ and $S+x$ instead of $ S \cup \{x\}$.} to the current solution $S$ is less than a $\nicefrac{1}{k}$-fraction of the current solution's value $f(S)$. 
    
    Moreover, after the arrival of each element $e_t$ and its possible addition to $S_t$, it %\eswapping 
    performs up to $N = \tilde{O}(\nicefrac{1}{\varepsilon})$ additional swaps. For this it considers all elements that have arrived so far, which we denote by $V_t$, %(and that are not part of $S_t$), 
    and it tries to add them one-by-one by the same condition and procedure that is used for newly arriving elements. 
    %(taking any element that satisfies the condition, as long as such an element exists). 
    The purpose of these extra swaps is to drive the maintained solution $S_t$ closer to a local optimum: a solution $S \subseteq V_t$ such that there is no element $x \in V_t$ such that $f(x \mid S) \geq \nicefrac{\phi}{k} \cdot f(S)$.
    %We say that a solution $S_t$ is a local optimum, if there exists no element $x \in V_t$, such that $f(x \mid S_t) \geq \nicefrac{\phi}{k} f(S_t)$.
    The improved approximation guarantee then stems from the fact that either the algorithm was at a local optimum in the not too distant past, or it performed many swaps since.

% Preliminaries
\section{Preliminaries}
\label{sec:preliminaries}

    We consider a set function $f: 2^V \to \bbRp$ on a ground set $V$ of cardinality $n$. Given two sets $X, Y \subseteq V$, the \emph{marginal gain} of $X$ with respect to $Y$, $\marginal{X}{Y}$, quantifies the change in value of adding $X$ to $Y$ and is defined as 
    \[
        \marginal{X}{Y} = f(X \cup Y) -  f(Y).
    \]
    When $X$ consists of a singleton $x$, we use the shorthand $f(x\mid Y)$ instead of $f(\{x\}\mid Y)$. Function $f$ is called \emph{monotone} if $\marginal{e}{X}  \geq 0$ for each set $X \subseteq V$ and element $e \in V$, and \emph{submodular} if for any two sets $X \subseteq Y \subseteq V$ and any element $e \in V \setminus Y$ we have $\marginal{e}{X} \ge \marginal{e}{Y}.$
    
    Throughout the paper, we assume that $f$ is monotone and that it is \emph{normalized}, i.e., $f(\emptyset) = 0$. We model access to the submodular function $f$ via a value oracle that computes $f(S)$ for given $S \subseteq V$. The problem of maximizing a function $f$ under a \emph{cardinality constraint} $k$ is defined as selecting a set $S \subseteq V$ with $|S| \le k$ that maximizes $f(S)$.

% Lower Bound
\section{Impossibility Result}
\label{sec:lower_bound}
    
    Putting computational efficiency aside, it may be possible to design a consistent algorithm which maintains the optimal solution, or an arbitrarily good approximation. We prove that this is not the case: no deterministic algorithm with constant consistency enjoys an approximation guarantee better than $2$. We remark that this is an information-theoretical bound, and concerns the streaming nature of the problem. 
        
    \begin{theorem}
    \label{thm:2_deterministic}
        Fix any constant $C$ and precision parameter $\e \in (0,1)$. No $C$-consistent (deterministic) algorithm provides a $(2-\e)$-approximation.  
    \end{theorem}
        \begin{proof}
            Fix any constant $C$, precision parameter $\e > 0$, and a deterministic algorithm $\cA$ that is $C$-consistent, we construct a covering instance such that $\cA$ does not maintain a $(2-\e)$ approximation.
            Let $G = \{g_1, \dots, g_n\}$ be a ground set and $V$ be a family of subsets of $G$ such that $V$ contains all the subsets of $G$ of cardinality $1$ and $k$, with $k = \nicefrac n2$. The covering function $f$ is naturally defined on $V$, and we consider the task of maximizing $f$ with cardinality $k$. 
            
            Observe the behaviour of $\cA$ on the sequence $\{g_1\}, \dots \{g_n\}$. At the end of this partial sequence $\cA$ maintains a certain solution $S = \{\{g_{i_1}\}, \dots, \{g_{i_{\ell}}\}\}$, with $\ell \le k.$ 
            %The next element arriving is exactly
            Now suppose the next element to arrive is $\{g_{i_1}, \dots, g_{i_{\ell}}, g_{i_{\ell+1}}, \dots, g_{i_k}\}$, where $g_{i_{\ell+1}},\dots, g_{i_k}$ are some arbitrary elements not covered by $S$. The value of the optimal solution after this insertion is $2k-1$ (just take the last subset and $k-1$ non overlapping singletons). The value of $S$ is $\ell \le k$ and, even if $\cA$ adds to $S$ the subset $\{g_{i_1}, \dots, g_{i_{\ell}}, g_{i_{\ell+1}}, \dots, g_{i_k}\}$ and $C-1$ other singletons, it cannot get a solution of value more than $k + C$. The theorem follows by choosing appropriate values for $k$: $k \ge \nicefrac {3C}{\e}$. 
        \end{proof}
        
% Window-Swapping

\section{\window} 
\label{sec:window}\label{sec:1consistent} 
    In this section, we present the \window algorithm, which achieves an approximation guarantee of $3.146+O(\nicefrac{1}{k})$ and $1$-consistency (changes at most one element for each insertion). \window maintains a benchmark set $B_t$ to which it adds all the elements that, upon arrival, exhibit a marginal contribution to $B_t$ that is at least $\nicefrac \beta k \cdot f(B_t)$. At any given stream operation $t$, the solution $S_t$ is given %after each insertion 
    by the last $k$ elements added to $B_t$. We refer to the pseudocode for further details. We prepare the analysis of the properties of $\window$ with two Lemmata. We start relating the value of the optimal solution with that of the benchmark. 
    
    \begin{algorithm}[t]
    \caption{\window} \label{alg:window-swapping}
    \begin{algorithmic}[1]
    \STATE \textbf{Environment:} Stream $V$, function $f$, cardinality $k$
    \STATE Threshold parameter $\beta \gets 1.14$
    \STATE $B_0 \gets \emptyset,$ $S_0 \gets \emptyset$, and $t \gets 1$ 
    \FOR{$e_t$ new element arriving}
        \IF{$f(e_t\mid B_{t-1}) \ge \frac \beta k f(B_{t-1}) $}\label{line:test}
            \STATE $B_t \gets B_{t-1} + e_t$
            \STATE $S_t \gets S_{t-1} + e_t$
            \IF{$|S_t| = k + 1$}
                \STATE remove from $S_t$ the element $e_s$ with smallest $s$
            \ENDIF
            \STATE $t \gets t+1$
        \ENDIF
    \ENDFOR
    \end{algorithmic}
    \end{algorithm}
    
    \begin{lemma}
    \label{lem:OPTvsV}
        After each insertion $e_t$, the following holds:
        \[
            f(\OPT_t) \le (1+\beta) \cdot f(B_t).
        \]
    \end{lemma}
    \begin{proof}
        Consider any element that belongs to $\OPT_t$ but not to the benchmark set $B_t$ after the computation following the insertion of $e_t$, i.e., $e_s \in \OPT_t \setminus B_t$, with $s \le t$. Element $e_s$ has not been included to $B_s$ (because it does not belong to $B_t \supseteq B_s$) upon its insertion, so the following holds:
        \begin{align}
            f(e_s\mid B_{t}) &\le f(e_s\mid B_{s-1}) \tag{by submodularity}\\
            &\le \tfrac \beta k f(B_{s-1}) \tag{since $e_s\notin B_s$ }\\
            &\le \tfrac \beta k f(B_t). \tag{by monotonicity}
        \end{align}
        So, for any element $e_s\in \OPT_t \setminus B_t$, it holds that
        \begin{equation}
        \label{eq:e_s}
            f(e_s\mid B_{t}) \le \tfrac \beta k f(B_t).    
        \end{equation}
        The above inequality is the crucial ingredient of the proof:
        \begin{align*}
            f(\OPT_t) &\le f(\OPT_t \cup B_t) \tag{by monotonicity}\\
            &\le f(B_t) + \sum_{e_s \in \OPT_t\setminus B_t} f(e_s\mid B_{t})\\
            &\le f(B_t) + |\OPT_t| \cdot \tfrac \beta k f(B_t) \tag{by  Ineq. \ref{eq:e_s}}\\
            &\le (1 + \beta)f(B_t).            \tag{because $|\OPT_t| \leq k$}
            %\tag{Because $\OPT_t$ is a solution}
        \end{align*}
        Note, the second inequality comes from submodularity.
    \end{proof}
    As a second preliminary step, we argue that the elements in $B_t$ that are not included to the current solution $S_t$ only account for a small fraction of $f(B_t)$.
    \begin{lemma}
    \label{lem:SvsV}
        After each insertion $e_t$, the following inequality holds:
        \[
            f(B_t) \ge \left(1+\frac \beta k\right)^k f(B_t \setminus S_t).
        \]
    \end{lemma}
    \begin{proof}
        The elements in the current solution are naturally sorted according to the order in which they are inserted in the stream and then added to the solution: $S_t = \{s_{t_1}, s_{t_2}, \dots s_{t_\ell}\}$. Element $s_{t_\ell}$ is the last one added and, clearly, $\ell \le k$ and $t_{\ell} \le t.$
        Each one of these $e_{t_i}$ elements has been added to the solution because it passed the value test: $
            f(s_{t_i} \mid B_{t_i-1})\ge \tfrac \beta k f(B_{t_i-1}). $
            
        Now, set $B_{t_i-1}$ can be rewritten in terms of the current benchmark set $B_t$ and the elements in the solution $S_t$: $B_{t_i-1} = B_t \setminus \{s_{t_i}, \dots, s_{t_\ell}\}$, so the previous inequality can be rewritten as
        \[
            f(s_{t_i} \mid B_t \setminus \{s_{t_i}, \dots, s_{t_\ell}\}) \ge  \tfrac \beta k f(B_t \setminus \{s_{t_i}, \dots, s_{t_\ell}\}).
        \]
        If we add to both sides of the above inequality the term ${f(B_t \setminus \{s_{t_i}, \dots, s_{t_\ell}\})}$, we get that
        \[
            f(B_t \setminus \{s_{t_i+1}, \dots, s_{t_\ell}\}) \ge \left(1+\frac \beta k\right) f(B_t \setminus \{s_{t_i}, \dots, s_{t_\ell}\}).
        \]
        Iterating the above argument we get the desired bound:
        \begin{align*}
            f(B_t) &\ge \left(1+\tfrac \beta k\right) f(B_t \setminus \{s_{t_\ell}\}) \ge \left(1+\tfrac \beta k\right)^2 f(B_t \setminus \{s_{t_{\ell-1}}, s_{t_\ell}\})\ge \dots \ge \left(1+\tfrac \beta k\right)^\ell f(B_t \setminus S_t).
        \end{align*}
        The Lemma follows by recalling that $\ell \le k$.
    \end{proof}
    We now have all the ingredients to analyze \window.
    
    \begin{theorem}
        \label{thm:1consistent}
        \window is $1$-consistent and maintains a $3.147 + O(\nicefrac 1k)$ approximation.
    \end{theorem}
    \begin{proof}
        First observe that the algorithm is indeed $1$-consistent: every time the solution $S_t$ changes, exactly one element is inserted and exactly one is removed from it.
        
        We move our attention to the approximation guarantee. We start by noting that
        \begin{align*}
            f(B_t) + \left(1+\tfrac \beta k\right)^k f(S_t)             &\ge \left(1+\tfrac \beta k\right)^k [f(S_t) + f(B_t \setminus S_t)] \tag{\Cref{lem:SvsV}}\\
            &\ge \left(1+\tfrac \beta k\right)^k f(B_t). \tag{by submodularity}
        \end{align*}
        By rearranging terms
        and applying \Cref{lem:OPTvsV} we get:
        \begin{align}
            f(S_t) &\ge \frac{\left(1+\frac \beta k\right)^k - 1}{\left(1+\frac \beta k\right)^k} f(B_t)
            \label{eq:multiplier}
            \ge \frac{\left(1+\frac \beta k\right)^k - 1}{\left(1+\frac \beta k\right)^k(1+\beta)} f(\OPT_t).
        \end{align}
        We conclude the proof by providing a general lower bound for the multiplier of the right-hand side of the last inequality. We know that the following simple chain of inequality holds:
        \[
            \left(1+\frac \beta k\right)^k \le e^\beta \le \left(1+\frac \beta k\right)^k \left(1-\frac {\beta^2} k\right)^{-1}
        \]
        Plugging the above inequality into the multiplier in \Cref{eq:multiplier}, we have
        \begin{align*}
            \frac{\left(1+\frac \beta k\right)^k - 1}{\left(1+\frac \beta k\right)^k(1+\beta)} \ge \frac{e^{\beta}-1}{e^{\beta}(1+\beta)} - \frac{\beta^2}{k(1+\beta)}\ge 0.3178 - \frac 1{k}. \tag{$\beta = 1.14$}
        \end{align*}
        Taking the inverse yields the desired factor.
    \end{proof}
% Eps-Swapping

\section[Epsilon-Swapping]{\eswapping}
\label{sec:eswapping}

    In this section we present and analyze the \eswapping algorithm, which exhibits a better approximation factor than both \swapping and \window. We refer to the pseudocode for further details. 
    There are two differences with respect to \window. First, the way in which elements in the solution are swapped out: it is not the ``oldest'' element to be removed, but one with small enough value. This is formalized in the routine \swap,
    which takes as input a set $S$ and an element $x$, and is responsible for inserting $x$ into $S$; if $S$ already contains $k$ elements, then $x$ is swapped with an element $r$ in $S$ with marginal value not larger than the average value of $S$ (so to maintain the cardinality of $S$ bounded by $k$). Note, such an element $r$  always exists by submodularity and a simple averaging argument:
    \[
        f(S) \ge \sum_{x \in S} f(x \mid S - x) \ge k \cdot \min_{x \in S} f(x \mid S -x)\,.
    \]
    The second difference is that after the arrival of each element and possibly its addition to the current solution, the algorithm performs up to $N \in \Tilde O(\nicefrac{1}{\e})$ additional swaps from $V_t$ into the solution, using the same rule and subroutine as for newly arriving elements. The additional swaps performed by \eswapping drive the maintained solution closer to a local optimum defined as follows.

    \begin{algorithm}[t]
    \caption{$\swap(S,x)$} \label{alg:swap}
    \begin{algorithmic}[1]
    \STATE \textbf{Input:} Set $S$ and element $x$
    \STATE \textbf{Environment:} Function $f$ and cardinality $k$
    \STATE \textbf{if} $|S| < k$, \textbf{then} \textbf{return} $S + x$
    \STATE Let $r \in S$ be any element s.t. $f(r \mid S - r) \le \nicefrac {f(S)}k $
    \STATE \textbf{return} $S - r + x$
    \end{algorithmic}
    \end{algorithm}
    
 \begin{algorithm}[t]
 \caption{\eswapping} 
    \begin{algorithmic}[1]
    \STATE \textbf{Input:} Precision parameter $\e$
    \STATE \textbf{Environment:} Stream $V$, function $f$, cardinality $k$
    \STATE $\phi \gets \frac{\sqrt 5 + 1}{2}$, $N \gets \lceil \frac 1 \e \log_\phi \frac {12}\e \rceil $
    \STATE $S_0 \gets \emptyset$ and $t \gets 1$
    \FOR{$e_t$ new element arriving}
        \IF{$f(e_t\mid S_{t-1}) \ge \frac \phi k f(S_{t-1}) $}\label{line:sc1}
            \STATE $S_t \gets \swap(S_{t-1}, e_t)$
        \ENDIF
        \FOR{$i = 1, \dots, N$}
            \IF{$\exists \, x \in V_t$ such that $f(x\mid S_{t}) \ge \frac \phi k f(S_{t}) $}\label{line:sc2}
                \STATE $S_t \gets \swap(S_{t}, x)$
            \ENDIF
        \ENDFOR
        \STATE $t \gets t + 1$
    \ENDFOR
    \end{algorithmic}
    \end{algorithm}
    
    \begin{definition}
    \label{def:local}
        We say that a dynamic solution $S_t$ is a local optimum if there exists no element $x$ in $V_t$ such that $f(x\mid S_t) \ge \tfrac \phi k f(S_t)$.
    \end{definition}
    The improved approximation guarantee stems from the fact that at any point in time, either the solution maintained by the algorithm was a local optimum not too far in the past, or many swaps were performed since.
    \begin{theorem}\label{thm:eswapping}
        \eswapping maintains a $(\phi + 1 + 9\e)$-approximation, where $\phi \approx 1.619$ is the golden ratio, and is $\tilde O(\nicefrac 1{\e})$-consistent.
    \end{theorem}
    
      Before proving the theorem, %main Theorem, 
      we introduce a notational convention. %we 
      %spend some words on a notation convention. 
      During the execution of the algorithm, elements may be added and removed multiple times from the dynamic solution. %from the solution, 
      Rather than thinking of such an element as one and the same element, it is convenient to think of this happening to multiple distinct copies of the same element so that each element is added and removed at most once.
       This allows us to work with sets instead of multi-sets in the analysis.

    \begin{proof}[Proof of \cref{thm:eswapping}]
        The bound on the consistency is immediate, as for each insertion there are at most $N+1 = \lceil \nicefrac 1{\e} \log_{\phi} \nicefrac {12}\e\rceil  +1 = \tilde O(\nicefrac{1}{\e})$ changes in the solution. The rest of the proof is devoted to the analysis of the approximation guarantee, which we prove by induction on the number of insertions. For the first element $e_1$ of the stream there is nothing to prove, as $S_1 = V_1 = \{e_1\}$. We analyze now the generic insertion $e_t$, with $t > 1$, assuming that the desired approximation holds for any previous insertion $s < t$. Let $t'$ be the last insertion index before $t$ in which the solution $S_{t'}$ was a local optimum (see \Cref{def:local}),
        % \textcolor{orange}{[Ashkan: should we define it here?]} \textcolor{blue}{[Federico: I added a pointer to the definition. Is this enough?]} \textcolor{red}{[Paul: Yes, I think so.]}
        and denote with $\tau$ the maximum between $t'$ and $(t - \lceil \e k \rceil)$. We remark that $t'$ is at least $1$, so $\tau$ is well defined. We have that $\OPT_t$ is the optimum after insertion $e_t$, and $\OPT_{\tau}$ is the optimum after insertion $e_{\tau}$. Sets $V_t$ and $V_{\tau}$, $V_t$ and $V_{\tau}$ are defined in a similar way. Consider how the solution changed between $S_{\tau}$ and $S_t$: some elements in $S_{\tau}$ were removed, some were added and remained in $S_{t}$, while others were added and later removed, possibly multiple times. To ease the analysis, we sort these inserted elements $s_1, s_2, \dots, s_{L}$ according to the order in which they were added (recall that multiple ``copies'' of the same element may appear in this sequence); this induces a natural sorting on the removed elements: we call $r_\ell$ the element that was swapped out to make room for $s_\ell$ (to avoid confusion, if no element was swapped out, we let $r_{\ell}$ be a dummy element with no value). We now define an auxiliary sequence of sets $A_\ell$ that interpolates between the solution at insertion $\tau$ and that at insertion $t$: $            A_\ell = S_{\tau} \cup \{s_1, \dots, s_{\ell}\} \setminus \{r_1, \dots, r_{\ell}\}.$
        
        It holds that $S_{\tau} = A_0$, while $S_t = A_L$. Moreover, the definition of the auxiliary sets motivates this relation:
        \begin{equation}
        \label{eq:relation}
            A_{\ell-1} + s_\ell = A_\ell + r_\ell.
        \end{equation}
        By a telescopic argument, the above relation and the design of \eswapping we have the following claim.
        \begin{claim}
        \label{cl:telescopic}
            The following inequality holds true:
            \[
                f(S_{t}) \ge f(S_{\tau}) + (\phi - 1) \sum_{\ell=1}^L f(s_\ell \mid A_{\ell-1}- r_\ell). 
            \]
        \end{claim}
        \begin{proof}[Proof of \Cref{cl:telescopic}]
            The change in value between two consecutive auxiliary sets can be decomposed as follows exploiting the relation in \Cref{eq:relation}: 
            \begin{equation}
            \label{eq:aux1}
                f(A_{\ell}) - f(A_{\ell-1}) = f(s_{\ell}\mid A_{\ell-1}) - f(r_\ell \mid A_{\ell}).
            \end{equation}
            Now, the marginal value of $s_\ell$ with respect to $A_{\ell-1}$ is at least $\nicefrac \phi k \cdot f(A_{\ell-1})$, by the swapping conditions in lines \ref{line:sc1} and \ref{line:sc2} of \eswapping. Furthermore, by the design of \swap, we know that the element $r_{\ell}$ that is removed to make room for $s_{\ell}$ has small value. In formula, 
            \begin{equation}
            \label{eq:aux2}
                f(s_{\ell}\mid A_{\ell-1}) \ge \tfrac{\phi}k f(A_{\ell-1}) \ge \phi f(r_{\ell} \mid A_{\ell-1}-r_{\ell}) 
            \end{equation}
            We can now prove directly the inequality in the statement:
            \begin{align*}
                f(S_t) &- f(S_\tau)\\ &=\sum_{\ell=1}^L f(A_\ell) - f(A_{\ell-1}) \tag{telescopic argument}\\&= \sum_{\ell=1}^L f(s_\ell\mid A_{\ell - 1}) - f(r_\ell \mid A_{\ell}) \tag{by Eqn. \ref{eq:aux1}}\\
                &\ge \sum_{\ell=1}^L f(s_\ell\mid A_{\ell - 1}) - f(r_\ell \mid A_{\ell-1} - r_\ell)\displaybreak[0]\\
                &\ge (\phi - 1) \sum_{\ell=1}^L f(r_\ell \mid A_{\ell-1}- r_\ell). \tag{by Eqn. \ref{eq:aux2}}
            \end{align*}
            Note, the second to last inequality follows by submodularity and the fact that $A^{\ell-1} - r_\ell = A_{\ell} -s_{\ell} \subseteq A_\ell $, due to the relation in \Cref{eq:relation}.
        \end{proof}
        
        Denote now with $I$ the set of elements that were inserted between $e_\tau$ and $e_t$ : $I = V_t \setminus V_\tau$, and with $A$ the set of all the elements that were, at some point, in the solution between time $\tau$ and $t$: $A = \cup_{\ell=\tau}^t A_{\ell}$. It is possible to relate the value of $S_t$ with that of the elements in $I$ and $A$:
        \begin{claim}
        \label{cl:local1}
            The following inequality holds true: 
            \[
                f(I \cup A) \le (1+4\e)f(S_t) + \sum_{\ell=1}^L f(r_\ell \mid A_{\ell-1}-r_\ell). 
            \]
        \end{claim}
        \begin{proof}[Proof of \Cref{cl:local1}]
            Consider any element $g$ in $I \cup A$. We have three cases: either element $g$ belongs to $S_t$, $g$ was added to the solution but was later swapped out, or it failed the swapping condition in line \ref{line:sc1} upon insertion. Now, sort these elements according to the order in which they were discarded by the algorithm: $(I \cup A) \setminus S_t = \{g_1, \dots, g_J\}$ ($g \in I \setminus A$ is discarded upon insertion, while $g \in A \setminus (I \cup S_t)$ is discarded when gets swapped out by the solution). For simplicity, denote with $G_j$ the set of the first $j-1$ such elements, we have the following two facts: (i) if $g_j \in I \setminus A$, then it means that $g_j = e_{t'}$ for some $t' \in \{\tau, \dots, t\}$, and the solution $S_{t'}\subseteq S_t \cup G_j$; (ii) if $g_j \in A \setminus (I \cup S_t)$, then it means that $g_j = r_\ell$ for some $\ell \in \{1, \dots, L\}$, and it holds that $A_{\ell} - r_{\ell}\subseteq S_t \cup G_j$.
            
            Exploiting these two facts and submodularity, we have the following chain of inequalities:
            \begin{align*}
                %New lines to make the inequalities clearer
                f(I \cup A) - f(S_t) &= \sum_{j=1}^J f(g_j \mid S_t \cup G_j)\\
                &\le \hspace*{-7pt}\sum_{e_{t'} \in I \setminus (S_t \cup A)} \hspace*{-10pt} f(e_{t'} \mid S_{{t'}-1}) + \sum_{\ell=1}^L f(r_\ell \mid A_{\ell-1}-r_\ell)
                \\
                &\le \frac \phi k \sum_{e_{t'} \in I \setminus (S_t \cup A)} \hspace*{-10pt} f(S_{{t'}-1}) +  \sum_{\ell=1}^L f(r_\ell \mid A_{\ell-1}-r_\ell)\\
                &\le 4\e f(S_t) +  \sum_{\ell=1}^L f(r_\ell \mid A_{\ell-1}-r_\ell).
            \end{align*}
            Note, the second inequality holds by the fact that $e_j$ failed the swapping condition in line \ref{line:sc1} upon insertion; while the third inequality follows by observing that the sequence of $f(S_{t'})$ is non-decreasing, there are at most $2 \e k$ elements in $I \setminus (S_t \cup A)$, and $\phi \in (1,2)$.
        \end{proof}
        Another useful property of the auxiliary sets $A_{\ell}$ is to provide a clean way to formalize that adding new elements to the solution multiplicatively improves the value of the solution.
        \begin{claim}
        \label{cl:multiplicative}
            The following inequality holds true:
            \[
                f(S_t) \ge \left(1 + \frac{\phi - 1}{k}\right)^L f(S_{\tau}).
            \]
        \end{claim}
        \begin{proof}[Proof of \Cref{cl:multiplicative}]
            Consider the generic subsequent terms $\ell-1$ and $\ell$, for $\ell = 1, \dots, L$. Starting from rearranging \Cref{eq:aux1}, we have the following: 
            \begin{align*}
                f(A_{\ell}) &= f(s_\ell \mid A_{\ell-1}) - f(r_{\ell} \mid A_{\ell}) + f(A_{\ell-1})\\
                &\ge f(s_\ell \mid A_{\ell-1}) - f(r_{\ell} \mid A_{\ell-1}-r_\ell) + f(A_{\ell-1})\\
                &\ge \frac{\phi-1}{\phi}f(s_{\ell} \mid A_{\ell-1}) + f(A_{\ell-1}) \tag{by Eqn. \ref{eq:aux2}}\\
                &\ge \left(1+\frac{\phi-1}{k}\right) f(A_{\ell-1}),
            \end{align*}
            where the first inequality follows by submodularity and the relation in \Cref{eq:relation}, while the last one by the design of \swap: an element is added to the solution only if its marginal contribution is at least a $\nicefrac \phi k$ fraction of $f(A_{\ell-1}).$
            Applying iteratively the above argument from $S_t = A_L$ to $S_{\tau} = A_0$ yields the desired result.
        \end{proof}
        We now have all the ingredients to directly address the crux of the proof.
        We have two cases we analyze separately: either $S_{\tau}$ is a local optimum, or it is not. 
        \paragraph{$S_\tau$ is a local optimum.} If $S_{\tau}$ is a local optimum, then all the elements in $\OPT_t$ that arrived before $e_{\tau}$, i.e., $\OPT_t \cap V_{\tau}$ have low marginal contribution with respect to $S_\tau$. Formally, we have the following result.
        \begin{claim}
        \label{cl:local2}
            If $S_{\tau}$ is a local optimum, then
            \[
                (1+4\e)f(S_t) + \sum_{\ell=1}^L f(r_\ell \mid A_{\ell-1}-r_\ell) \ge f(\OPT_t) - \phi f(S_\tau).
            \]
        \end{claim}
        \begin{proof}[Proof of \Cref{cl:local2}]
             To prove this result, it suffices to argue that the right-hand side of the inequality in the statement is at most $f(I \cup A)$, as it is then possible to conclude the argument by combining it with \Cref{cl:local1}. We have the following:
             \begin{align*}
                 f(\OPT_t) &- f(I \cup A)\\
                 &\le f(\OPT_t \mid I \cup A) \tag{by monotonicity} \\
                 &=  f(\OPT_t \cap V_{\tau} \mid I \cup A) \tag{since $I = V_t \setminus V_{\tau}$} \\
                 &\le  \sum_{e \in \OPT_t \cap V_{\tau}} f(e \mid S_{\tau})\\
                 &\le  \phi f(A_{\tau})\tag{$A_\tau$ local optimum}.
             \end{align*}
             Note, the second inequality holds by submodularity as $S_{\tau}$ is contained into $A$. Reordering the terms of the inequality we get the desired lower bound on $f(I \cup A).$
        \end{proof}
        Summing the inequality in \Cref{cl:local2} with $\phi$ times the inequality in \Cref{cl:telescopic} yields the desired bound, thus concluding the argument for the first case:  
         \begin{align*}
            (1+\phi + 4\e) f(S_t)
            &\ge f(\OPT_t) + [\phi (\phi-1) - 1] \sum_{\ell=1}^L f(r_\ell \mid A_{\ell-1}-r_\ell) = f(\OPT_t).
        \end{align*}
        In the previous inequality we crucially used the definition of the golden ratio as the solution of $\phi^2 - \phi - 1 = 0.$

        \paragraph{$S_{\tau}$ is not a local optimum.} If $S_{\tau}$ is not a local optimum, then it means that $L$, the total number of swaps between insertion $e_\tau$ and $e_t$, is at least $\e \cdot k \cdot N$, where $N$ is defined in the pseudocode as $\lceil \nicefrac 1\e \cdot \log_\phi \nicefrac {12}\e\rceil$. If we complement this with \Cref{cl:multiplicative} we get:
        \begin{align*}
            f(S_t) &\ge  f(S_{\tau}) \cdot \left(1 + \frac{\phi - 1}{k}\right)^{k \log_\phi \nicefrac {12}\e} 
            % \tag{by \Cref{cl:multiplicative} and definition of $N$}
            \\
            &\ge \frac{{12}}{\e} f(S_{\tau})\tag{because $(1 + \tfrac xn)^n \ge 1+x$ }\\
            &\ge \frac{{12}}{(1+\phi+9\e)\e} f(\OPT_\tau),
        \end{align*}
        where in the last inequality we crucially used the inductive assumption. By rearranging and using that $\e \in (0,1)$ and $\phi \in (1,2)$, we get the following simple relation, which proves that the value of the elements arrived up to time $\tau$ can be safely ignored:
        \begin{equation}
            \label{eq:non_local}
            f(\OPT_\tau) \le \e f(S_t).
        \end{equation}
        We have all the ingredient to deal with the last case: 
        \begin{align*}
            f(\OPT_t) 
            &\le f(\OPT_t \cap V_{\tau}) + f(I) \tag{by submodularity}\\
            &\le f(\OPT_\tau) + f(I) \tag{by optimality of $\OPT_\tau$}\\
            &\le \e f(S_t) + f(I) \tag{by Equation \ref{eq:non_local}}\\
            &\le \e f(S_t) + f(I \cup A) \tag{by monotonicity}\\
            &\le (1+ 5\e) f(S_t) + \sum_{\ell=1}^L f(r_\ell \mid A_{\ell-1}-r_\ell) \tag{by \Cref{cl:local1}}\\
            &\le \left(1 + \frac{1}{\phi-1} + 9\e \right) f(S_t) \tag{by \Cref{cl:telescopic}}\\
            &= (1+\phi + 9 \e) f(S_t),
        \end{align*}
        where in the last equality we used the definition of $\phi$ as the golden ratio. This last case concludes the proof.
    \end{proof}
    
% Experiments
\section{Experiments}
\label{sec:experiments}

    In this section we evaluate the performance of our two
    algorithms on real-world data sets\footnote{The code of the experiments is available at \url{https://github.com/fedefusco/Consistent-Submodular}.}. We report here three case
    studies, while we defer to Appendix C other (qualitatively
    analogous) results, as well as further implementation details. 
    We present additional experimental results that illustrate the gains in worst-case approximation guarantee in \cref{swap:tight}.
    As benchmarks we consider the \swapping algorithm which provides a $4$-approximation and is $1$-consistent and the \sieve algorithm, a $(2+\e)$-approximation that is not consistent (see \Cref{app:instability} for further details on the instability of the algorithm).
    
    \begin{figure*}
    \begin{subfigure}{.32\textwidth}
      \centering
      \includegraphics[width=\linewidth]{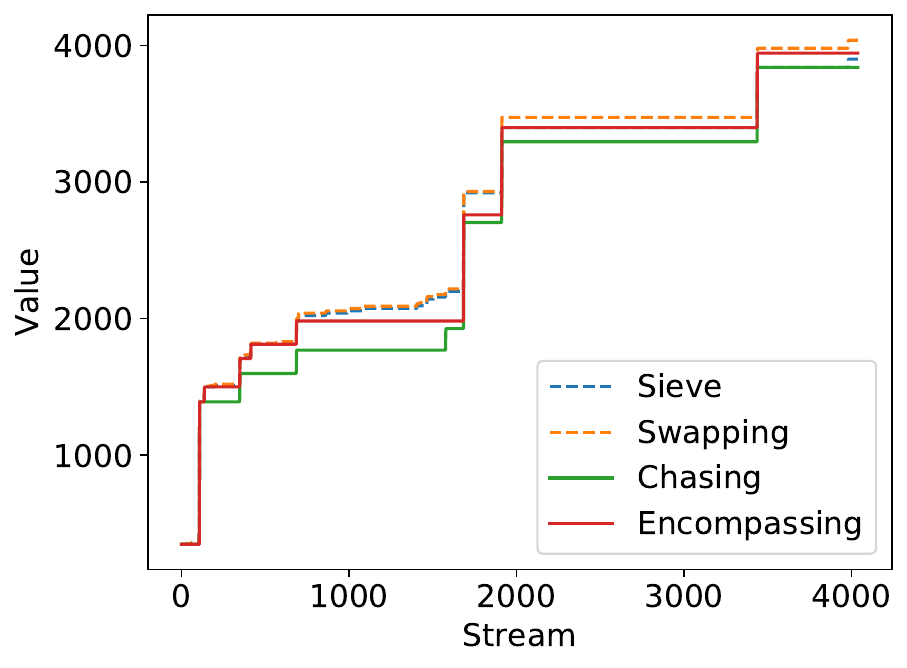}
      \caption{Influence Maximization on Facebook}
      \label{fig:fb-objective}
    \end{subfigure}
    \begin{subfigure}{.32\textwidth}
      \centering
      \includegraphics[width=0.96\linewidth]{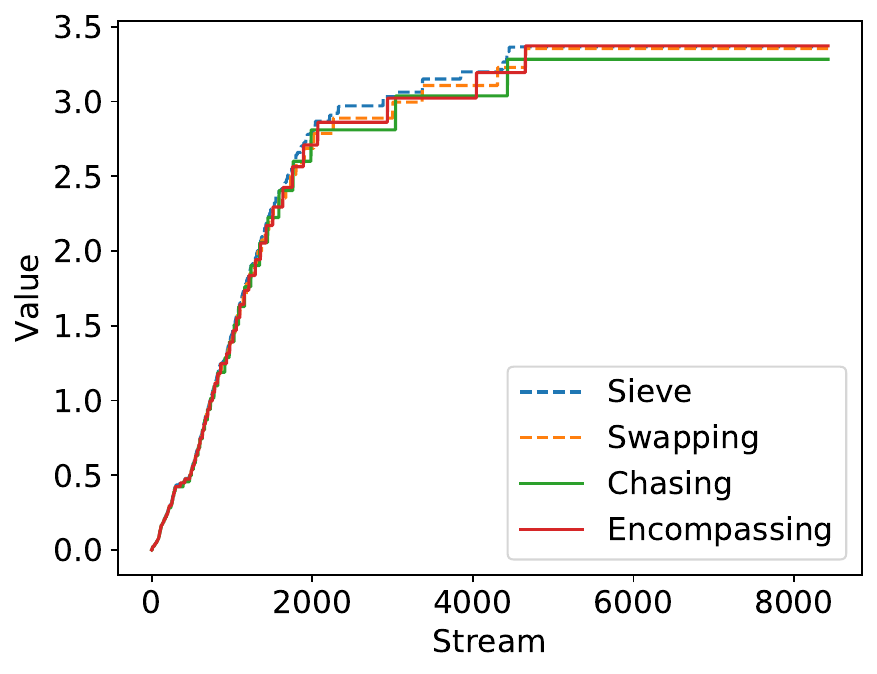}
      \caption{k-medoid Clustering on RunInRome}
      \label{fig:rome-kmedoid-objective}
    \end{subfigure}
    \begin{subfigure}{.32\textwidth}
      \centering
      \includegraphics[width=0.95\linewidth]{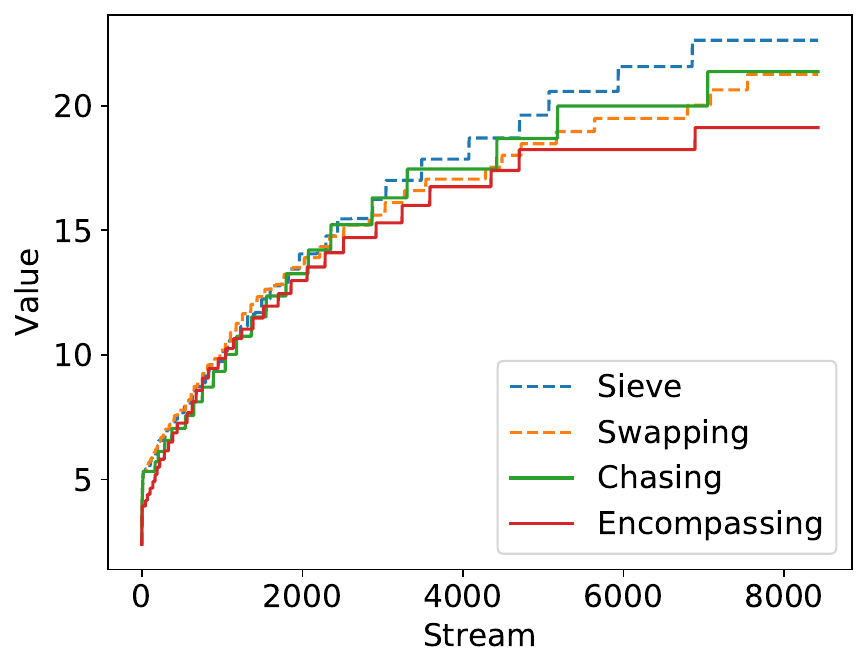}
      \caption{LogDet Maximization on RunInRome}
      \label{fig:rome-logdet-objective}
    \end{subfigure}
    \begin{subfigure}{.32\textwidth}
      \centering
      \includegraphics[width=\linewidth]{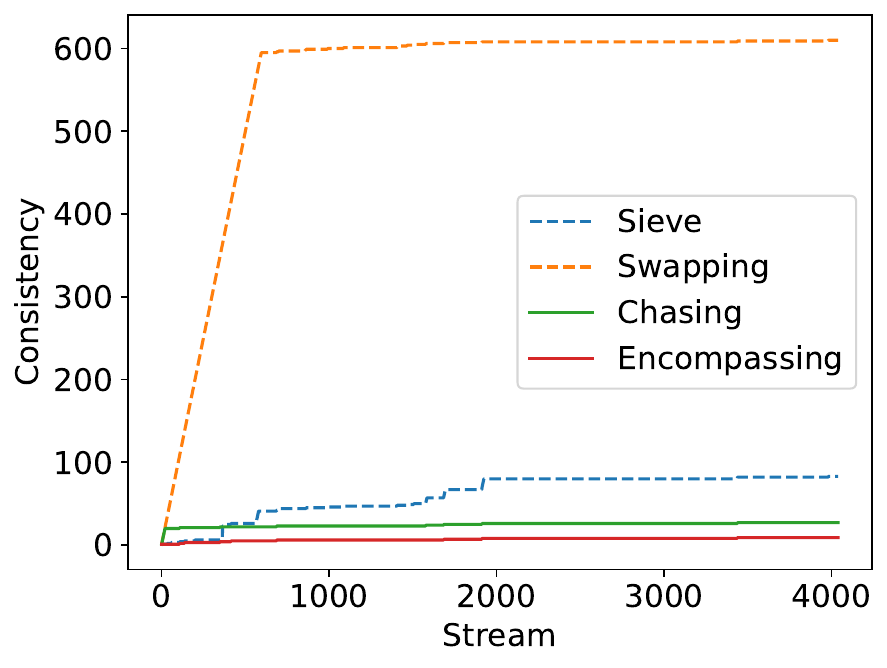}
      \caption{Influence Maximization on Facebook}
      \label{fig:fb-consistency}
    \end{subfigure}%
    \begin{subfigure}{.32\textwidth}
      \centering
      \hspace*{0.2cm}
      \includegraphics[width=0.98\linewidth]{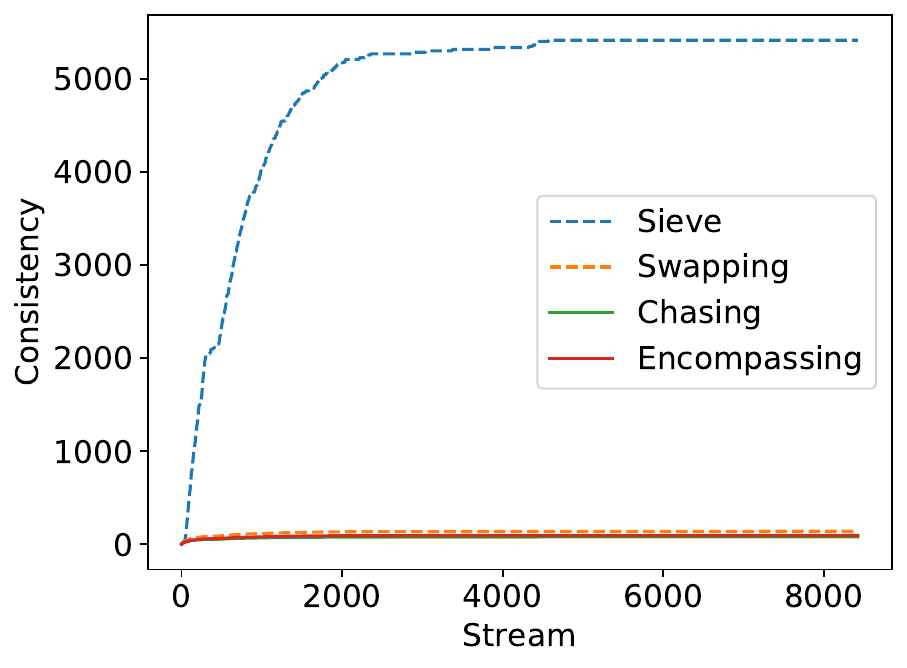}
      \caption{k-medoid Clustering on RunInRome}
      \label{fig:rome-kmedoid-consistency}
    \end{subfigure}
    \begin{subfigure}{.32\textwidth}
      \centering
      \hspace*{0.5cm}
      \includegraphics[width=0.99\linewidth]{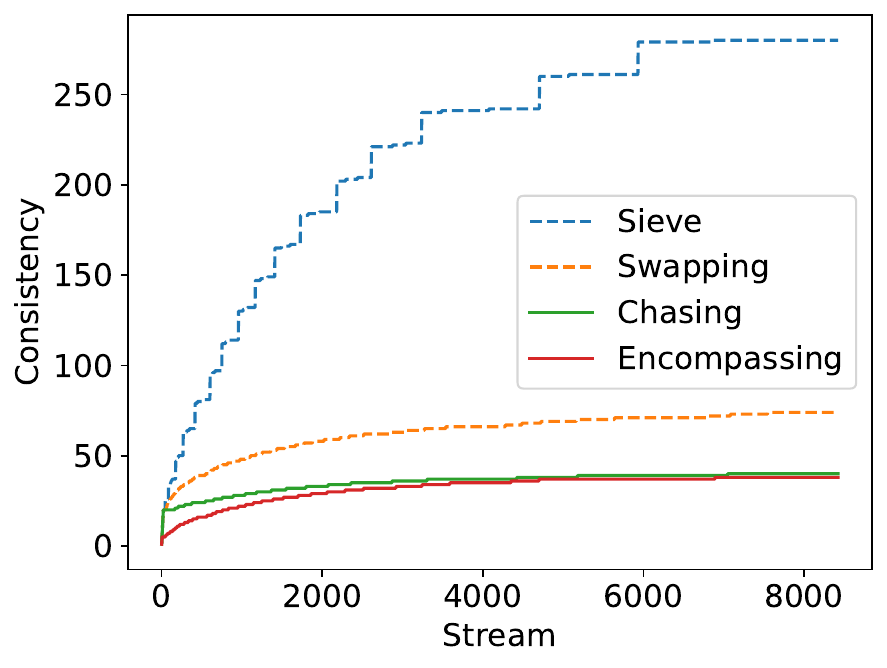}
      \caption{LogDet Maximization on RunInRome}
      \label{fig:rome-logdet-consistency}
    \end{subfigure}
    \caption{Experimental Results. The first row reports the objective values, the second one the cumulative consistency.}
    \label{fig:results}
    \end{figure*}
    
    \paragraph{Influence Maximization.}
        For our first case study we consider the problem of influence maximization on a social network graph \citep[e.g.,][]{Norouzi-FardTMZ18,HalabiMNTT20}, where the goal is to maintain a subset of the nodes to ``influence'' the rest of the graph. In such application, consistency is crucial as changing nodes may entail costs relative to terminating and issuing new contracts. We use the Facebook dataset from \citet{McAuleyL12} that consists of $4039$ nodes $V$ and $88234$ edges $E$ and, as measure of influence we consider the monotone and submodular dominating function:
        \[
        f(S) = |\{v \in V: \exists s \in S \text{ and } (s,v) \in E\}|.
        \]
                
    \paragraph{Summarizing Geolocation Data.}
        Our second and third case study concern the problem of maintaining a stable and representative summary from a sequence of geographical coordinates \citep[e.g.,][]{MirzasoleimanK017,DuettingFLNZ22}. We use the RunInRome dataset \citep{Federico22}, that contains $8425$ positions recorded by running activity in Rome, Italy. We consider two different objective functions used in geographical data summarization: the $k$-medoid and the kernel log-det. 
        Consider the $k$-medoid function on the metric set $(V,d)$ 
        $
            L(S) = \frac{1}{|V|} \sum_{v \in V} \min_{e \in S} d(e,v).
        $
        By introducing an auxiliary point $e_0 \in V$ we can turn $L$ into a monotone submodular function \citep{MirzasoleimanKSK13}: 
        \[
            f(S) = L(e_0) - L(S + e_0). 
        \]
        In our experiment we set $e_0$ to be the first point of each dataset.
        For the second objective, consider a kernel matrix $K$ that depends on the pair-wise distances of the points, i.e. $ K_{i,j} = \exp\{-\frac{d(i,j)^2}{h^2}\}$ 
        where $d(i,j)$ denotes the distance between the $i^{th}$ and the $j^{th}$ point in the dataset and $h$ is some constant. Following \citet{Krause14}, another common monotone submodular objective is $f(S) = \log \det (I + \alpha K_{S,S})$, where $I$ is the $|S|-$dimensional identity matrix, $K_{S,S}$ is the principal sub-matrix corresponding to the entries in $S$, and $\alpha$ is a regularization parameter (that we set to $10$ in the experiments).
        
        \paragraph{Experimental Results.}
        In \cref{fig:results}, we present the performance of our algorithms and the benchmarks. The first row (\Cref{fig:fb-objective,fig:rome-kmedoid-objective,fig:rome-logdet-objective}) features the objective value of the dynamic solution maintained by the algorithms, while the second row (\Cref{fig:fb-consistency,fig:rome-kmedoid-consistency,fig:rome-logdet-consistency}) reports the cumulative number of changes in the solutions. The experiments show that our algorithms, \window and \eswapping, 
        achieve comparable value as \swapping and \sieve; while achieving notable savings in the total number of changes. %(typically up to 50\% with respect to \swapping and often by various order of magnitude with respect to \sieve). 
        {For instance, in the setting of \cref{fig:fb-consistency}, \swapping is significantly less consistent on aggregate than our algorithms (around a factor $25$), while \sieve changes the solution about $3-4$ times more often. The superior cumulative consistency of our algorithms is also clear in the other experiments; in the settings of  \cref{fig:rome-kmedoid-consistency} and \cref{fig:rome-logdet-consistency} \sieve performs order of magnitudes more changes than either of our algorithms (about 500x and 10x), while \swapping performs between 50\% and 100\% more.} 
        The strict ``insertion rules'' implemented by our two algorithms seem to 
        guarantee that only the crucial elements of the dataset are added to the solution. This phenomenon empirically induces a desirable global stability over the entire stream -- which goes beyond the theoretical per-round guarantees -- at the cost of possibly discarding moderately good elements.
    
% Conclusions
\section{Conclusion}

    In this paper, we initiate the study of consistency in submodular maximization. Consistency is a natural measure of stability of the online solution maintained by an algorithm, and has been extensively studied for clustering, facility location and online learning. 
    We present two consistent algorithms, \window and \eswapping, that exhibit a different approximation-consistency trade off ($3.147 + O(\nicefrac{1}{k})$ and $1$-consistent vs. $2.619 + O(\varepsilon)$ and $O(\nicefrac 1\e)$-consistent). They both substantially improve on the state of the art (a consistent $4$-approximation), moving the approximability boundary closer to the optimal approximation factor, as evidenced by the information-theoretical lower bound of $2$ that we prove to hold for any consistent deterministic algorithm. 
    Besides closing the remaining gap in the approximation factor, our work raises many natural and compelling questions. 
    First, the investigation of randomized
    algorithms may lead to better results, even beyond the lower
    bound of $2$. Second, while some known algorithms already
    exhibit consistency, the explicit study of consistent algorithms for possibly non-monotone submodular functions
    and more general constraints (e.g., matroids and knapsack)
    may lead to improved results.

\section*{Acknowledgments}

Federico Fusco is supported by the FAIR (Future Artificial Intelligence Research) project
PE0000013, funded by the NextGenerationEU program within the PNRR-PE-AI scheme (M4C2, investment 1.3, line on Artificial Intelligence), the ERC Advanced
Grant 788893 AMDROMA ``Algorithmic and Mechanism Design Research in Online Markets'', the PNRR MUR project IR0000013-SoBigData.it, and the MUR PRIN project ``Learning in Markets and Society''.

\bibliographystyle{plainnat}
\bibliography{bibliography}

%%%%%%%%%%%%%%%%%%%%%%%%%%%%%%%%%%%%%%%%%%%%%%%%%%%%%%%%%%%%%%%%%%%%%%%%%%%%%%%
%%%%%%%%%%%%%%%%%%%%%%%%%%%%%%%%%%%%%%%%%%%%%%%%%%%%%%%%%%%%%%%%%%%%%%%%%%%%%%%
% APPENDIX
%%%%%%%%%%%%%%%%%%%%%%%%%%%%%%%%%%%%%%%%%%%%%%%%%%%%%%%%%%%%%%%%%%%%%%%%%%%%%%%
%%%%%%%%%%%%%%%%%%%%%%%%%%%%%%%%%%%%%%%%%%%%%%%%%%%%%%%%%%%%%%%%%%%%%%%%%%%%%%%
\appendix

\section{Instability of known algorithms}
\label{app:instability}
    
    We propose here two instances that highlight the instability of known algorithms. The instance in \Cref{ex:covering} is such that both the optimal solution and the output of the greedy algorithm \citep{nemhauser1978analysis} change entirely after every insertion. We then briefly discuss, in \Cref{ex:exponential}, a simple instance that forces the \sieve algorithm \citep{BadanidiyuruMKK14} and its modified version \sieveplus \citep{kazemi2019submodular} to behave in a non-consistent way.
    \begin{example}
    \label{ex:covering}
        Let $\delta \in (0,1)$ be a small parameter used to break ties, and consider the following weighted covering instance, parameterized by an integer $i$ and cardinality constraint $k$. The base set $E$ is given by the pairs $\{(a,b), \text{ for } a,b \in \{0, \dots, i\}\}$. We refer to each pair $(a,b)$ as an item. The weights of the items are as follows: all items have unitary weight, but the following:
        \[
            \begin{cases}
                w_{(0,0)} =& 0\\
                w_{(a,0)} =& \delta \cdot  (2a + 1) \quad \text{ for } a \neq 0\\
                w_{(0,b)} =& \delta \cdot 2b \quad \text{ for } b \neq 0
            \end{cases}
        \]
        The weighted covering function is monotone submodular and is defined as follows: \[
            f(S) = \sum_{(a,b): \exists s \in S, (a,b) \in s} w_{(a,b)}.
        \]
        Note, $f$ is defined over subsets of $E$, not on items. The subsets of $E$ we consider in our instance are the rows and columns of $E$: $R_a$ is defined as $\{(a,0), \dots, (a,i)\}$, while $C_b$ is defined as $\{(0,b), \dots, (i,b)\}$. The stream is constructed as follows: $C_1, R_1, C_2, R_2, \dots, C_{\ell}, R_{\ell}, \dots, C_i, R_i$. Consider now what happens after $2k$ insertions. The optimal solution (which is the same output by running greedy on the elements arrived so far) is as follows: if the last arrived element is a row, then the optimal solution is given by the last $k$ arrived rows; conversely, if the last arrived element is a column, then the optimal solution is given by the last $k$ arrived columns. This means that the $k$ elements in the dynamic solution change after each insertion! Note, the elements in the first row and first column are only there for tie-breaking.
    \end{example}
    
    \begin{example}
    \label{ex:exponential}
        The \sieve algorithms lazily maintains a set of geometrically increasing active thresholds $O$ (of the type $\tau = (1+\e)^j$, for some $j \in \mathbb{Z}$ and input parameter $\e$) and a candidate solution for each one of them; then outputs the best of these candidates. In particular, when a new element $e_t$ arrives, with value way larger than all the previous ones, a new threshold is activated and the corresponding candidate solution $S_{\tau}$ is initiated ($S_{\tau} = \{e_t\}$). It is then clear that any instance characterized by elements with dramatically increasing values would force the algorithm to continuously change its solution. For instance, consider an additive function with $f(e_t) = 2^t$: after each insertion, the solution output by \sieve would be the singleton $\{e_t\}$. Playing with similar arguments, it is not hard to construct an instance that completely change solution every $k$ insertions (e.g., $f(e_t) = k^{\lceil \nicefrac tk \rceil}$).
    \end{example}

\section{The analysis of \swapping is tight}\label{swap:tight}

    The \swapping algorithm is known to provide a $4$-approximation to the optimum \citep{ChakrabartiK15}. In this Section we first report the pseudocode for completeness, and then prove that the analysis is tight, meaning that for any $\e \in (0,1)$, there exists an instance of the problem where the solution computed by \swapping is at least a $(4-\e)$ factor away from the optimal one (\Cref{ex:swapping}). Finally, in \Cref{fig:synthetic} we report the empirical performances of \swapping, \sieve, and our algorithms on such hard instance. 
    \begin{algorithm}[t!]
	\caption{\swapping}
	\begin{algorithmic}[1]
		\STATE \textbf{Environment:} stream $\pi$ of elements, function $f$, cardinality $k$
		\STATE $S \gets \emptyset$
		\FOR{each new arriving element $e$ from $\pi$}
		    \STATE $w(e) \gets f(e\mid S)$
		    \IF{$|S| < k$}
		        \STATE $S \gets S + e$
		    \ELSE 
                \STATE $s_e \gets \argmin\{w(y) \mid y \in S\}$
                \IF{$2 \cdot w(s_e) \le w(e)$}
                    \STATE $S \gets S - s_e + e$
                \ENDIF
		    \ENDIF
		\ENDFOR
		\STATE \textbf{Return} $S$
	\end{algorithmic}
    \end{algorithm}
    
    \begin{example}
    \label{ex:swapping}
    Fix any $\e \in (0,1)$, and consider the following weighted covering instance, parameterized by an integer $i$ that we set later and the cardinality constraint $k = 2^i$. 
    The set of items is $E = \{e_\ell^j \mid j \in \{0,\ldots,i\}, \ell\in \{1,\ldots,k\}\}$.
    Consider the partition of $E$ into $i+1$ bundles of items $E^0, E^1, \dots, E^i$, where each bundle has $k$ items $E^j = \{e_1^j, \dots, e_k^j\}$. Let $\delta > 0$ be a small positive constant, which we will set later. The weight of the generic element $e_{\ell}^j$ in $E$ is $w_{\ell}^j = 2^{j}$ if $j \neq i$ and $w_{\ell}^i = 2^{i} - \delta$ otherwise. 
    Now that we have the auxiliary set $E$, we can define the stream $\pi$ of subsets of $E$ as follows. For $0 \leq j < i$, let $\pi^j$ be the subsequence $\{e_1^j\}, \dots, \{e_k^j\},E^j$. Let $\pi^i$ be the subsequence $\{e_1^i\}, \dots, \{e_k^i\}$ (without bundle $E^i$ at the end). Then $\pi$ is given by the concatenation of $\pi^0, \pi^1, \dots, \pi^i$.

    Now, the behaviour of \swapping on $\pi$ is clear: 
    it maintains in the solution the last $k$ singletons that arrived up to bundle $E^{i-1}$ and ignores the elements in $E^i$ (because of the small $\delta$). In particular, at the end of the stream outputs the solution $S = \{\{e_1^{i-1}\},\dots, \{e_k^{i-1}\}\}$, for a value of 
    \[  
        f(S)=\sum_{\ell=1}^k w_\ell^{i-1} = k \cdot 2^{i-1}.
    \]
    \end{example}

    Consider now the optimal solution $S^{\star}$ given by the $i$ bundles $E^0, \dots, E^{i-1}$ and $k-i$ singletons from the last bundle, e.g., $\{e^{i}_1\}, \dots, \{e_{k-i}^i\}$, for a value of
    \[
        f(S^{\star}) = \sum_{j=0}^{i} \sum_{\ell=1}^k w_\ell^{j} - \sum_{\ell=k-i+1}^{k}w_\ell^{i} \ge k \cdot (2^{i+1}-1) -k \delta - i\cdot2^{i}.
    \]
    Note, $S^{\star}$ is indeed the optimal solution because of our choice of $k=2^i$: the total weight of the elements in $E_0$ is $k$, while a singleton from $E_i$ has weight $2^i - \delta$. We can now focus on the approximation factor, we have: 
    % \textcolor{red}{[Paul: Changed the $-i$ in the bracket to $+i$].}
    \[
        \frac{f(S^\star)}{f(S)} \ge 4 - \frac{2}{2^i}(\delta+1+i).
    \]
    Now, the negative terms go to zero when $i$ goes to infinity (and $\delta$ is small enough), thus for any fixed precision $\e$ it is possible to set $i$ and $\delta$ so that $\nicefrac{f(S^\star)}{f(S)} \ge 4 - \e.$
    
    \begin{figure}
    \begin{subfigure}{.4\textwidth}
      \centering
      \includegraphics[width=\linewidth]{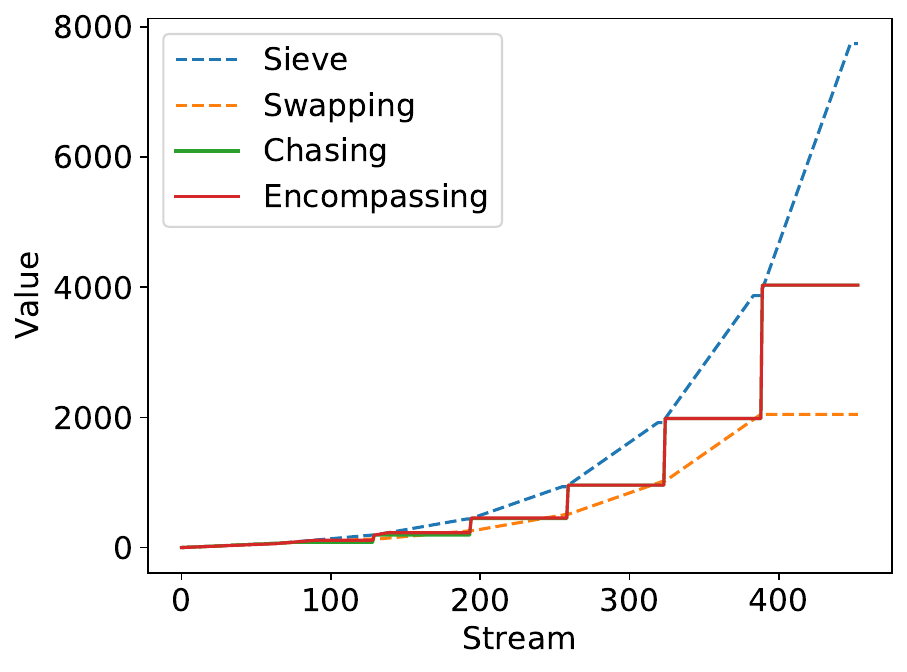}
      \caption{Weighted Covering - Objective Value}
    \end{subfigure}%
    \hfill
    \begin{subfigure}{.4\textwidth}
      \centering
      \includegraphics[width=\linewidth]{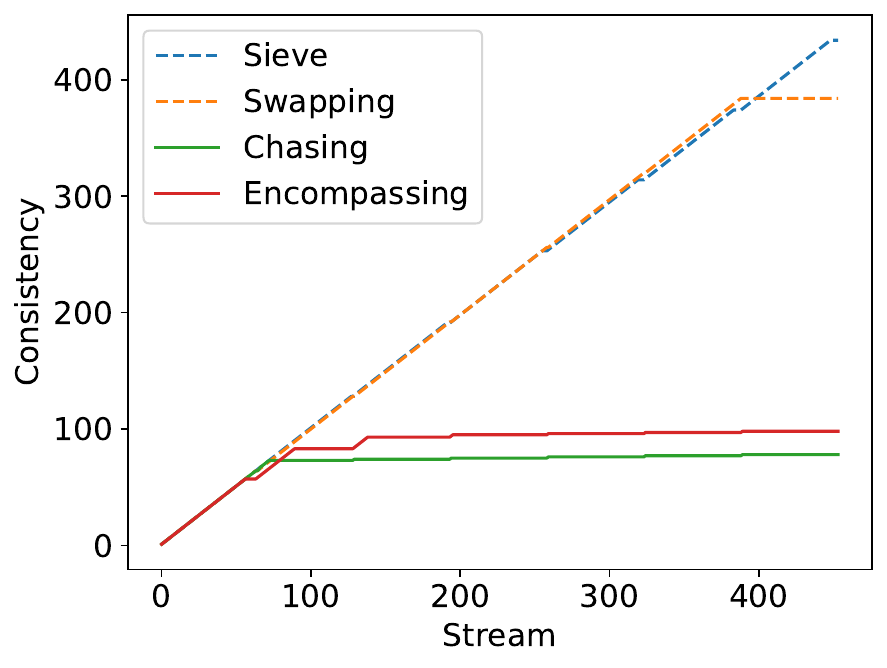}
      \caption{Weighted Covering - Cumulative Consistency}
    \end{subfigure}
    \caption{Experimental results on the weighted covering instance of \Cref{ex:swapping}, for $\delta = 0.01$, $i = 7$, and $\e = 0.1$.}
    \label{fig:synthetic} 
    \end{figure}
    
\section{Further Experimental Results}
\label{app:other_experiments}

    In our experiments, we set $\e = 0.1$ in \sieve and \eswapping, while the cardinality constraint $k$ is consistently set to $20.$ The order of the stream of elements is the one intrinsic in the dataset we consider. In \Cref{fig:other_results}, we report three extra experimental case studies. Besides studying the k-medoid and logdet objective on a random sample ($10332$ points) from the Uber pickups dataset \citep{Uber14} (see the last two columns of \Cref{fig:other_results} for the results), we present results for Personalized Movie Recommendation (first column of \Cref{fig:other_results}).

\paragraph{Personalized Movie Recommendation.}
        Movie recommendation systems are one of the common experiments in the context of submodular maximization  \citep[e.g.,][]{AmanatidisFLLMR21,DuettingFLNZ22,HalabiFNTT23}. In this experiment, we have a large collection $M$ of movies that arrive online and we want to design a recommendation system that proposes movies to users. For example, the summary may be a carousel of `recommended movies’ presented to a downstream user, and we would like the selection to be fairly stable.
        We use the MovieLens 1M database \citep{movielens16}, that contains 1000209 ratings for 3900 movies by 6040 users. Based on the ratings, it is possible to associate to each movie $m$, respectively user $u$, a feature vector $v_m$, respectively  $v_u$.
        More specifically, we complete the users-movies rating matrix and then extract the feature vectors using a singular value decomposition and retaining the first $30$ singular values \citep{TroyanskayaCSBHTBA01}. Following the literature \citep[e.g.,][]{MitrovicBNTC17}, we measure  the quality of a set of movies $S$ with respect to user $u$ (identified by her feature vector $v_u$), using the following monotone submodular objective function:
        \[
            f_u(S) = (1-\alpha) \sum_{s \in S} \scalar{v_u}{v_s}_+ + \alpha \cdot \sum_{m \in M} \max_{s \in S} \scalar{v_m}{v_s},
        \]
        where $\scalar{a}{b}_+$ denotes the positive part of the scalar product.
        The first term is linear and sum the {\em predicted scores} of user $u$ (that is chosen as a random point in $[0,1]^{30}$ in our experiments) for the movies in $S$, while the second term has a facility-location structure and is a proxy for how well $S$ {\em covers} all the movies. Finally, parameter $\alpha$ balances the trade off between the two terms; in our experiments it is set to $0.95$.

    \begin{figure}
    \begin{subfigure}{.32\textwidth}
      \centering
      \includegraphics[width=\linewidth]{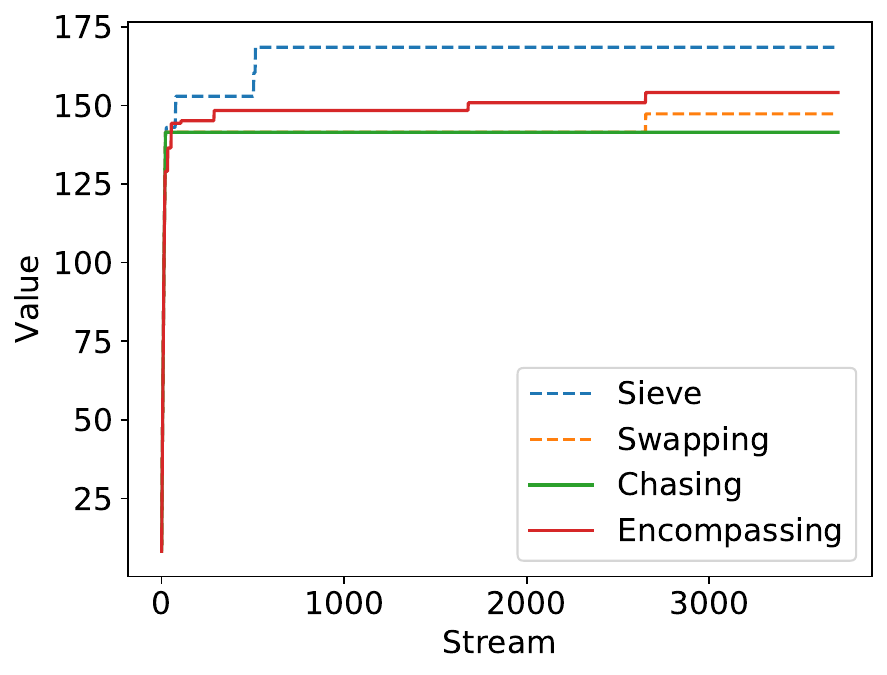}
      \caption{MovieLens Dataset}
      \label{fig:movie-objective}
    \end{subfigure}
     \begin{subfigure}{.32\textwidth}
      \centering
      \includegraphics[width=\linewidth]{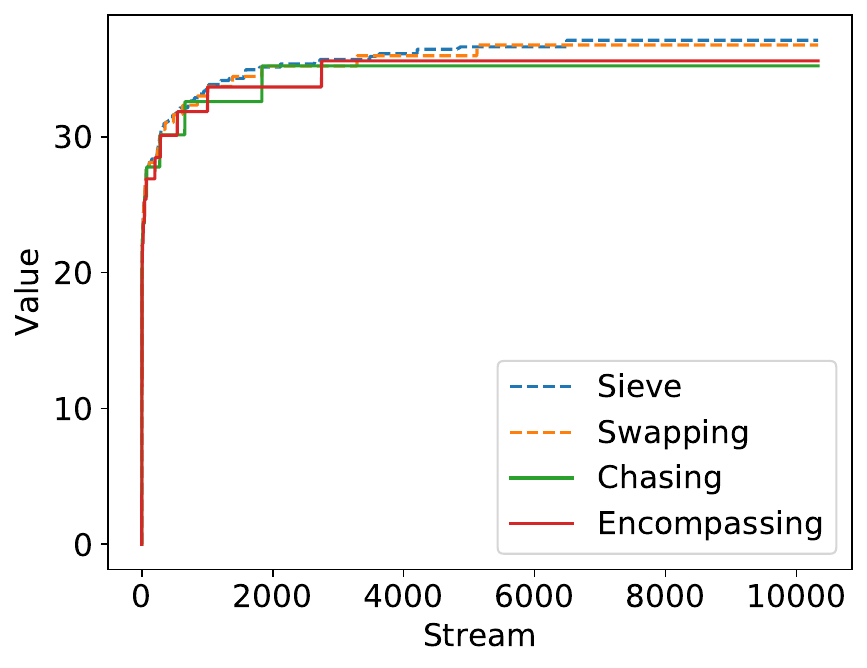}
      \caption{k-medoid Clustering on Uber Dataset}
      \label{fig:uber-kmedoid-objective}
    \end{subfigure}
    \begin{subfigure}{.32\textwidth}
      \centering
      \includegraphics[width=\linewidth]{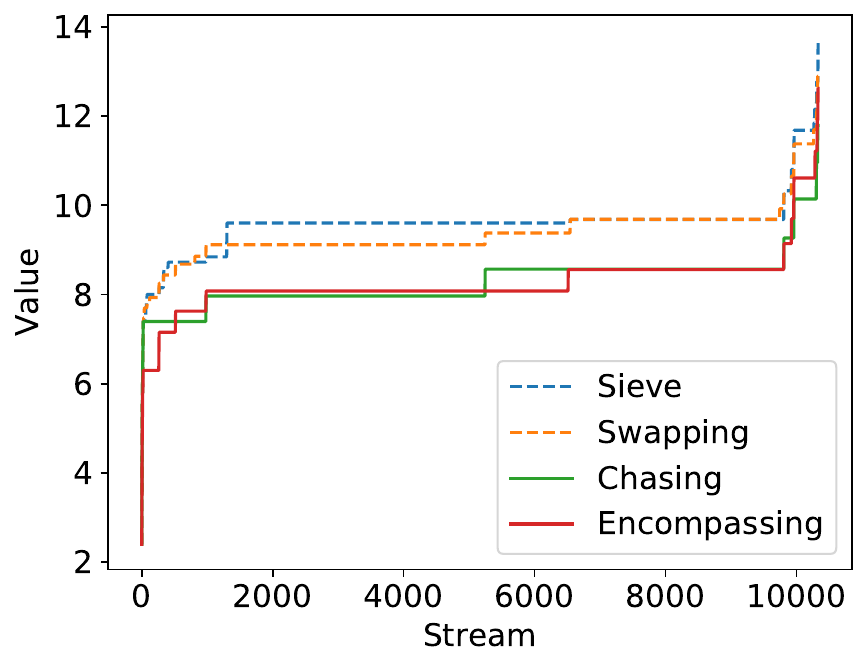}
      \caption{LogDet Maximization on Uber Dataset}
      \label{fig:uber-logdet-objective}
    \end{subfigure}
    \begin{subfigure}{.32\textwidth}
      \centering
      \includegraphics[width=\linewidth]{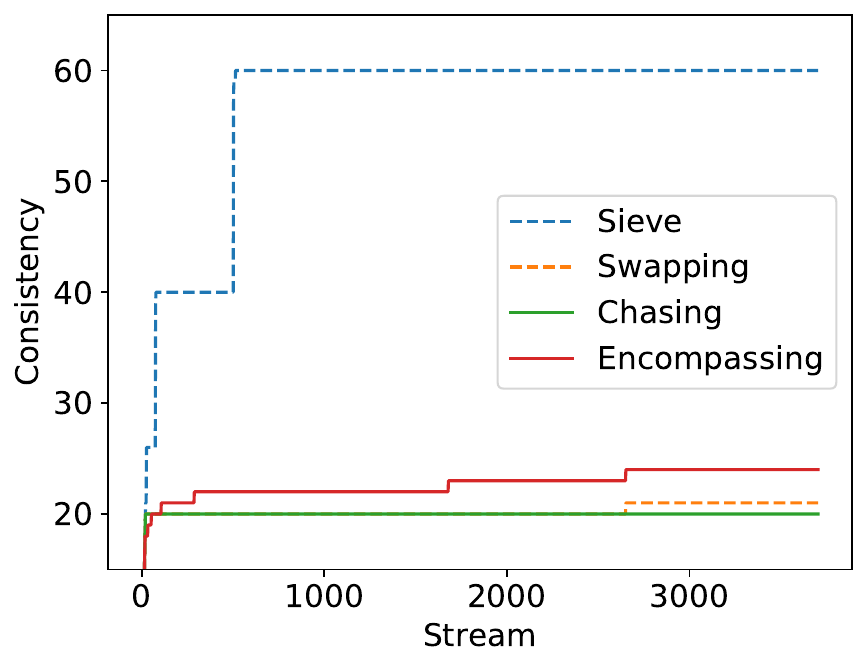}
      \caption{MovieLens Dataset}
      \label{fig:movie-consistency}
    \end{subfigure}
    \begin{subfigure}{.32\textwidth}
      \centering
      \includegraphics[width=\linewidth]{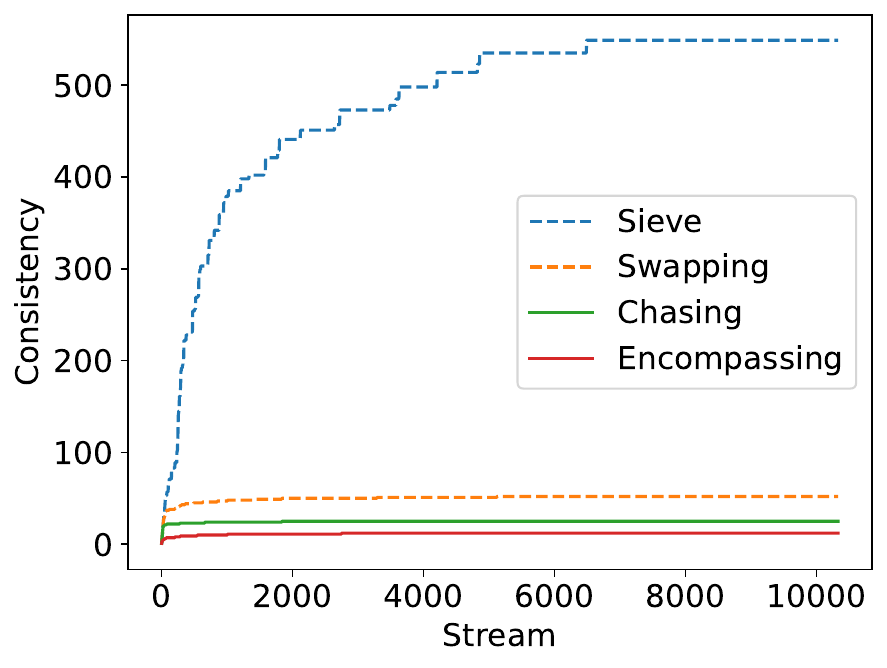}
      \caption{k-medoid Clustering on Uber Dataset}
      \label{fig:uber-kmedoid-consistency}
    \end{subfigure}
    \begin{subfigure}{.32\textwidth}
      \centering
      \includegraphics[width=\linewidth]{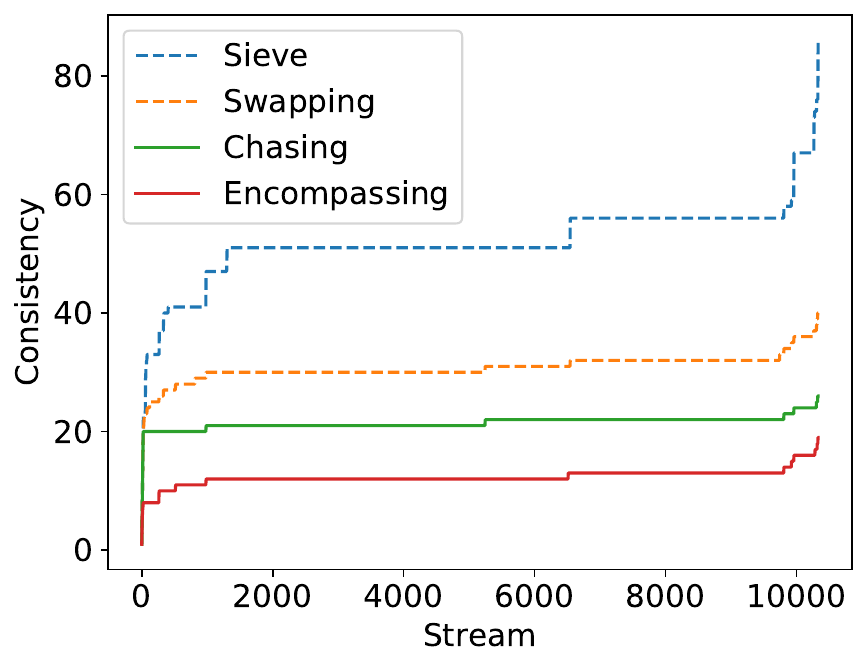}
      \caption{LogDet Maximization on Uber Dataset}
      \label{fig:uber-logdet-consistency}
    \end{subfigure}
\caption{Further Experimental Results. %Experimental Results.
The first row reports the objective values, the second one the cumulative consistency.}
\label{fig:other_results}
\end{figure}
\end{document}